\documentclass[graybox]{svmult}

\usepackage{mathptmx}       
\usepackage{helvet}         
\usepackage{courier}        
\usepackage{type1cm}        
\usepackage{makeidx}         
\usepackage{graphicx}        
\usepackage{multicol}        
\usepackage[bottom]{footmisc}
\usepackage{multirow}%
\usepackage{amsmath,amssymb,amsfonts}%
\usepackage{amsthm}%
\begin{document}

\raggedbottom

\newcommand{\s}{\sim}
\newcommand{\e}{\varepsilon}
\newcommand{\p}{\varphi}
\newcommand{\qm}{q_-}
\newcommand{\qp}{q_+}
\newcommand{\LR}{\quad\Longrightarrow\quad}
\newcommand{\half}{\tfrac{1}{2}}

\title*{Connecting the Deep Quench Obstacle Problem with Surface Diffusion via their Steady States\\[1ex]}
\titlerunning{The Deep Quench Obstacle Problem}

\author{Eric A. Carlen, Amy Novick-Cohen, and Lydia Peres Hari}
\institute{Eric A. Carlen \at Department of Mathematics, Hill Center, Rutgers University, 110 Frelinghuysen
Road, Piscataway, 08854-8019, New Jersey, {USA}, \email{carlen@math.rutgers.edu}
\and Amy Novick-Cohen \at Department of Mathematics, Technion-IIT, Haifa 32000, Israel, \email{amync@technion.ac.il}
\and Lydia Peres Hari \at Department of Mathematics, Technion-IIT, Haifa 32000, Israel, \email{lydia@technion.ac.il}}
%
%
%
%
%
%
%


\subtitle{\textit{In memory of Maria Conceicao Carvalho and her passion for life and science.}}
\maketitle

\abstract{In modeling phase transitions, it is useful to be able to connect diffuse interface descriptions of the dynamics with corresponding limiting sharp interface motions.
In the case of the deep quench obstacle problem (DQOP) and surface diffusion (SD), while a formal connection was demonstrated many years ago, rigorous proof of the connection has yet to be established.
In the present note, we show how information regarding the steady states for both these motions can provide insight into the dynamic connection, and we outline  tools that should
enable further progress. For simplicity, we take both motions to be defined on a planar disk.}





\section{Introduction}\label{sec1}
Many two component mixtures  exist stably at one temperature, but become unstable at lower temperatures; the subsequent instability typically initiates  phase separation,
leading to the appearance of spatial regions characterized by two different compositional phases. During the early stages of phase separation, the distinction between
the two phases is not sharp and a diffuse description is appropriate. As phase separation progresses, the phases become more distinct,
and a sharp interface description is appropriate. Accordingly, in physically realistic models, it should be possible to pass from one description to the other.
Unfortunately, often  there is  a gap between what is physically reasonable and what is possible to justify with mathematical rigor.
For example, though the diffuse interface Cahn-Hilliard model with a logarithmic potential and a degenerate mobility has been shown using formal asymptotics  \cite{CENC} to yield
the sharp interface surface diffusion model, the connection has yet to have been made rigorous.  In the present note, we focus on the zero temperature limit of the
Cahn-Hilliard model with a logarithmic potential and degenerate mobility, namely on the deep quench obstacle problem with degenerate mobility, which we shall subsequently refer to here simply as the deep quench obstacle problem (DQOP), and
its connection with motion by (isotropic) surface diffusion (SD).

After recalling below some relevant background with regard to both models, in  Section \ref{Steadystates} we  discuss
the steady states in some detail for both models, focusing in particular on the minimum energy steady states, and then in Section \ref{section3} we outline certain tools that we are using to bridge the
two evolutions.

\bigskip
Let us consider both motions, (DQOP) and (SD), to be defined  in $\Omega \subset R^2,$  where $\Omega$ is a disk centered at the origin whose  radius,  $R_0$, is $O(1)$. In considering time evolution, we set $\Omega_T:=\Omega \times (0,T)$ and $\partial\Omega_T:=\partial\Omega \times (0, T)$ with $0<T \le \infty$.
The  deep quench obstacle problem \cite{OonoPuri} with the degenerate mobility $M(u)=1 - u^2$, can be expressed as
\begin{equation} \nonumber
(\rm{DQOP})\quad \left\{  \begin{array}l
u_t= \nabla \cdot  M(u) \nabla w, \quad
w+u + \epsilon^2 \triangle u \in \partial I_{[-1,1]}(u), \quad (x,t)\in \Omega_T, \\[1ex]
n \cdot \nabla u= n \cdot M(u) \nabla w=0,  \quad (x,t)\in \partial\Omega_T, \\[1ex]
u(x,0)=u_0(x), \quad x \in \Omega, \end{array} \right.
\end{equation}
where  $\partial I_{[-1,1]}(\cdot)$ denotes the sub-differential of the indicator function $I_{[-1,1]}(\cdot)$ and
$n$ denotes the unit exterior normal to $\partial\Omega$.
The deep quench obstacle problem, (DQOP), constitutes the formal zero temperature ($\Theta \downarrow 0$) limit of the
Cahn-Hilliard equation with degenerate mobility and with a logarithmic potential \cite{ElliottLuckhaus}:
\begin{equation} \nonumber\label{e:dsCH}
\left\{  \begin{array}l
u_t= \nabla \cdot  M(u) \nabla w, \quad
w=\frac{\Theta}{2}\{ \ln(1+u) - \ln(1-u) \} -u - \epsilon^2 \triangle u, \quad (x,t)\in \Omega_T, \\[1ex]
n \cdot \nabla u= n \cdot M(u) \nabla w=0,  \quad (x,t)\in \partial\Omega_T, \\[1ex]
u(x,0)=u_0(x), \quad x \in \Omega. \end{array} \right.
\end{equation}
Given a smoothly embedded curve $\Gamma=\Gamma(t) \subset \Omega$, the curve $\Gamma(t)$ is said to evolve by motion by surface diffusion  \cite{Mullins1957} if, up to  rescaling by constants,
 \begin{equation} \nonumber
(\rm{SD})\quad \left\{  \begin{array}l
V= - \kappa_{s s}, \quad (x(s,t),t)\in \Omega_T,\\[1ex]
\Gamma(0)=\Gamma_0, \end{array} \right.
\end{equation}
where $s$ denotes an arc-length parametrization of $\Gamma(t)$, and $V$ and $\kappa$ denote, respectively, the normal velocity and the mean curvature of $\Gamma(t),$
defined in accordance with the exterior normal, $n$, relative to the arc-length parametrization, $s$.

Often it is convenient to express the deep quench obstacle problem with degenerate mobility somewhat informally  as
\begin{equation} \label{DQOP}
\quad \left\{  \begin{array}l
u_t= - \nabla \cdot  M(u) \nabla\,(u + \epsilon^2 \triangle u), \quad -1 \le u \le 1, \quad (x,t)\in \Omega_T, \\[1ex]
n \cdot \nabla u= n \cdot M(u) \nabla \,(u + \epsilon^2 \triangle u)=0,  \quad (x,t)\in \partial\Omega_T, \\[1ex]
u(x,0)=u_0(x), \quad x \in \Omega. \end{array} \right.
\end{equation}
As shown in \cite{CENC}, by considering  (DQOP) on the slow time scale, $\tau=\epsilon^2 t$,  to leading order the $\epsilon \downarrow 0$ limit of (DQOP) yields surface diffusion motion (SD)
for $\Gamma=\Gamma(t)$, where $\Gamma(t)$ denotes  the limit of the  $O(\epsilon)$ width interfaces which  partition the composition $u(x,t)$ in the context of (DQOP) into two phases, $u=\pm 1$, in $\Omega$. Off hand, $\Gamma(t)$ may contain one or more  components.

\bigskip
With regard to existence, the following is implied by
\cite[Theorem 1.1]{BNCN}:

\smallskip\par\noindent
\begin{theorem} \label{oldreg}
Let $\varepsilon>0$ and $T>0$, and let $\Omega \subset R^2$ be a bounded disk centered at the origin.
Let $( \cdot , \cdot )$ denote the $L^2(\Omega)$ inner product, and let $\langle \cdot, \cdot \rangle$ denote the $H^1(\Omega),$ $(H^1(\Omega))'$ duality pairing.
Suppose  that $u_0 \in \mathcal{K}:= \{ \eta \in H^1(\Omega)\,  | \, |\eta| \le 1\}.$ Then there exists a pair $\{ u,w \},$ such that
$u \in L^2(0,T; H^2(\Omega)) \cap H^1(0, T; (H^1(\Omega))') \cap L^{\infty}(0, T; \mathcal{K} ),$ $w \in L^2(\Omega_T),$ with $w \in H^1_{loc}(\{M(u)>0\})$ for a.e. $t \in (0,T)$,
 and
\begin{equation}
\left\{
\begin{array}l
\langle\frac{\partial u}{\partial t} , \eta \rangle + \int_{ \{M(u)>0\}} \nabla w \cdot M(u) \nabla \eta \, dx =0, \quad \forall \eta \in H^1(\Omega) \;\;\; a.e.  \;\;\; t \in (0,T),\\[1ex]
\epsilon^2(\nabla u, \nabla \eta - \nabla u) - (u, \eta -u) \ge (w, \eta -u), \quad \forall \eta \in \mathcal{K} \;\;\; a.e. \;\;\; t \in (0, T),\\[1ex]
u(x,0)=u_0, \quad x \in \Omega.
\end{array} \right. \end{equation}
\end{theorem}

With regard to (SD),  we  follow the discussion in \cite{Wheeler2013}.
Let $\Gamma: R \rightarrow R^2$ be a regular smooth immersed plane curve, which is  periodic and closed with period $P \in (0, \infty)$,
so that in fact $\Gamma: S^1 \rightarrow R^2.$  We shall assume throughout that $\Gamma$ is parameterized by arc-length.  With regard to local
existence,

\smallskip\par\noindent
\begin{theorem} \label{SDexistlocal}
Suppose  $\Gamma_0: R \rightarrow R^2$ is a regular closed periodic  curve parametrized by arc-length, of class $\mathcal{C}^2 \cap W^{2,2}$
with $||\kappa||_{L^2(\Gamma_0)} < \infty$. Then there exists  $T \in (0, \infty]$ and a unique one-parameter family of immersions parametrised by arc-length, $\Gamma: R \times [0, T) \rightarrow R^2,$
 such that
$(i) \; \Gamma(0,\cdot)=\Gamma_0$,
$(ii) \; V=-\kappa_{ss}$,
$(iii) \; \Gamma(\cdot,\, t) \hbox{\, is of class \,} \mathcal{C}^\infty$ and periodic with period $|\Gamma(\cdot,\,t)|,\; \forall  t \in (0, T)$, and
$(iv) \; T \hbox{\, is maximal.\,}$
\end{theorem}
\par\noindent
The uniqueness mentioned  in Theorem \ref{SDexistlocal} is modulo rotations, translations, changes in orientation, in accordance with the natural group of invariances for geometric flows in general, and for (SD) in particular.

Though the theorem above is formulated for one regular immersed circular curve, clearly Theorem \ref{SDexistlocal} readily generalizes to accommodate a finite union $\cup_{i\in I} \Gamma_i,$  $I \subset \mathcal{N},$  of such curves, which suffices
for the purpose of the discussion that follows.

\bigskip
Note that while the existence for (DQOP) is guaranteed by Theorem \ref{oldreg}   for arbitrary $T>0$,  existence for (SD) is guaranteed by Theorem \ref{SDexistlocal} only on some maximal interval. Indeed in the context
of (SD) the maximal interval may well be finite for various reasons,  for example, due to finite time self-intersection or curvature singularity formation.  If we wish to consider and compare the steady states for
both evolutions, the following theorem  which prescribes conditions  guaranteeing global existence for (SD)  is helpful, and can be readily adapted for $\cup_{i\in I} \Gamma_i,$ $I \subset \mathcal{N}$.

\smallskip\par
\begin{theorem}(see \cite[Theorem 1.1]{Wheeler2013})\label{SDglobal}
Let $\Gamma_0:  S^1 \rightarrow R^2$ be a regular smooth immersed closed curve with finite enclosed signed area, $\mathcal{A}(\Gamma_0)>0$ with $\int_{\Gamma_0} \kappa \,ds=2\pi$. Then there exists a constant
$K^\ast>0$, such that if
$$ K_{\rm{osc}}(\Gamma_0)< K^\ast \hbox{\, and \,} I(\Gamma_0) < \exp({K^\ast}/({8 \pi^2})),$$
where $K_{\rm{osc}}(\Gamma)= |\Gamma| \int_\Gamma (\kappa-\bar{\kappa})^2 \, ds$, with $\bar{\kappa} = |\Gamma|^{-1} \int_\Gamma \, \kappa \, ds$, denotes the normalized oscillation of the curvature
and $I=|\Gamma|^2/(4 \pi \mathcal{A}(\Gamma))$ is the isoperimetric ratio,
 the (SD) evolution for $\Gamma:S^1 \times [0, T) \rightarrow R^2$, with $\Gamma_0$ as initial data, exists for all time and converges exponentially fast to a round circle with radius $\sqrt{\mathcal{A}(\Gamma_0)/\pi}$.
\end{theorem}

\bigskip
\begin{remark} \label{remark1}
We remark that  more regularity is required for the initial conditions in Theorems \ref{SDexistlocal} and Theorem \ref{SDglobal} than appears to be  required
for the initial conditions in Theorem \ref{oldreg}. However   in connecting the (SD) flow with the (DQOP) flow, the (SD) curve
 $\Gamma(t)$ corresponds rather naturally to the (DQOP) level set, $\{ u(x,t)=0 \,| \, (x,t) \in \Omega_T \,\},$ whose support lies within the set $\{ (x,t) \in \Omega_T \,|\, \mathcal{M}(u(x,t))>0\,\}$
 where the required regularity is guaranteed for $T>0$.
 \end{remark}

\bigskip
With regard to prominent dynamical features
for (DQOP) and (SD), notably
\begin{equation} \label{funcDQOP}
 E(t):= \frac{1}{\epsilon |\Omega|} \int_{\Omega} \{ (1-u^2) + \epsilon^2 |\nabla u|^2 \} \, dx \quad\rm{\, and \,}\quad \mathcal{L}(t):=|\Gamma(t)|,
  \end{equation}
 where $E(t)$ is a scaled free energy\footnote{Here $E(t)$ has been scaled so that typically as equilibrium is approached, $E(t) \propto \frac{{L}(t)}{|\Omega|}$, where
 $L(t)$ reflects the length of the interface of between the two phases following phase separation; see \cite{BNCN}.}
  for (DQOP),  are monotonically non-increasing  for (DQOP) and (SD), respectively. Furthermore,
 \begin{equation} \label{volDQOP}
 \bar{u}(t):=\frac{1}{|\Omega|} \int_{\Omega} u(x,t) \, dx \quad\rm{\, and \,}\quad \mathcal{A}(t)= - \frac{1}{2}\int_0^{\mathcal{L}(t)} \Gamma(s,t) n \, ds, \end{equation}
 where $\bar{u}(t)$, the mean mass of $u(x,t)$,  and $\mathcal{A}(t)$, the signed  area enclosed  by $\Gamma(t)$, are invariant under the respective (DQOP)  and (SD)   evolutions.
 The  Gamma limit ($\epsilon \downarrow 0$)  of the constrained mean mass ($\bar{u}=\bar{u}_0$) minimizers of $E(t)$  is well known \cite{PSternberg88,JerrardSternbergJEMS} to yield a curve $\Gamma$ with prescribed
  enclosed signed area, $\mathcal{A}_0=\mathcal{A}(0)$; this can be readily be shown
 to hold also for $u$ constrained to lie in $\mathcal{K}$ \cite{NC_Meccanica,PSternberg88}, with
  \begin{equation} \label{eqnAu}
  \mathcal{A}_0=\frac{1}{2}|\Omega|(1 - \bar{u}_0),
    \end{equation}
  where $\bar{u}_0:= \bar{u}(0)$  in the context of our geometric and boundary assumptions.
Notably, (DQOP)  can be formulated
 as a conserved $H^{-1}$ gradient flow  with respect to the energy functional, $E(t)$, and (SD) can be formulated as an $H^{-1}$ gradient flow  with respect to $\mathcal{L}(t)$; see \cite{TaylorCahn} and the discussion in Section \ref{section3}.
 Moreover, the functionals
 $$\textrm{Ent}(t):=\frac{1}{|\Omega|} \int_{\Omega} \{ (1-u)\ln(1-u) + (1+u)\ln(1+u) \} \, dx,$$ $$K_{\rm{osc}}(\Gamma(t))=|\Gamma| \int_\Gamma (k-\bar{k})^2 ds,$$
are non-increasing along the respective (DQOP) and (SD) flows\footnote{Here $\textrm{Ent}(t)$ reflects the physical entropy of the  system,
while $E(t)$ is a (scaled) free energy.}.
 For (DQOP), for initial conditions $u_0 \in \mathcal{K}$ which correspond to  a perturbation of $u_0 \equiv \bar{u}$, $\bar{u} \in (-1,1)$, the dynamics can be characterized in terms of
 an initial regime of linear instability and a long time coarsening regime, \cite{KohnOtto,DCDS_NCS2009,JSP_NC2010}. For further discussion, see \cite{BNCN}.


%
\section{Steady states}\label{Steadystates}

\bigskip
We assume both motions to be defined within $\Omega \subset R^2$, where $\Omega$ is a planar disc centered at the origin with  radius  $R_0$, where $R_0$ is $O(1)$. When considering time evolution, we set $\Omega_T=\Omega \times (0,T),$ where $0<T \le \infty$.

\bigskip
Note that in  (DQOP), no flux and Neumann boundary conditions are implied. For simplicity, we shall henceforth  assume, more specifically, that $u \equiv -1$ in a $\delta$-neighborhood of $\partial \Omega$,
with $\epsilon \ll \delta \ll O(1)$. While this assumption implies the boundary conditions given in (DQOP), it  somewhat limits the resultant dynamics and steady states.
Similarly, in studying (SD), surface diffusion motions, within $\Omega$, we  consider  evolving curves $\Gamma(t) \subset \Omega$ which, more specifically, satisfy
$\Gamma(t) \subset \mathcal{B}_{R_0}(\delta)$, where $\mathcal{B}_{R_0}(\delta)$ refers to the open planar disc centered at the origin with radius $R_0$ which has a bounding annular neighborhood with width $\delta$
where $u \equiv -1$.

\subsection{Steady states and energy minimizers for (SD)}\label{subsect2.1}

In considering (SD) energy minimizers, by recalling (\ref{funcDQOP}) one may view the energy to be given by the length of the curve, $\mathcal{L}(t)$,  to within scaling constants. We allow $\Gamma(t)$
to be comprised of a finite number of nonintersecting components, $\Gamma(t) = \cup_{i \in I} \Gamma_i(t)$, $I \subset \mathcal{N}$ with  $\Gamma_i \cap \Gamma_j = \emptyset$ for $i \neq j$.
Accordingly, the evolution by (SD) of each component $\Gamma_i$, $ i \in I$, may be prescribed as
\begin{equation} \label{SDi}
V^i = - \kappa^i_{s_i s_i}, \end{equation}
where $V^i$ and $\kappa^i$  denote respectively, the normal velocity and  the mean curvature of $\Gamma_i(t)$, and $s_i$ denotes an arc-length parametrization of  $\Gamma_i(t)$.

Within the context of these assumptions, it readily follows from (\ref{SDi}) that steady states correspond of a finite disjoint union of circular curves. Since the area encompassed by
curves evolving by surface diffusion is conserved under the evolution \cite{Wheeler2013},
it follows that  $\sum_{i \in I} A(\Gamma_i(t))$ is a conserved quantity. Taking into account the isoperimetric
inequality and (\ref{SDi}), we may now conclude the following:

\smallskip\par
\begin{theorem} \label{ssSD}
Under the assumptions outlined above,
the set of minimum energy steady states for (SD) corresponds to  the set of circular curves with radius $(A^*/\pi)^{1/2}$, with $A^*=\sum_{i \in I} A(\Gamma_i(0))$, which are located somewhere within $\mathcal{B}_{R_0}(\delta)$.
Any finite union of disjoint circular curves such that $\sum_{i \in I} A(\Gamma_i(t))=\sum_{i \in I} A(\Gamma_i(0))$ also corresponds to a steady state for (SD).
\end{theorem}

\subsection{Steady states and energy minimizers for (DQOP)}\label{subsect2.2}

The energy minimizing steady states for (DQOP) correspond to the energy minimizers of $E(t)$ within the set  $u \in \mathcal{K}$ and which satisfy $\bar{u}=\bar{u}_0.$
Given the definition of the energy $E(t)$ and the geometry described at the beginning of this section, considerations of energy symmetrization and regularity
 \cite{BrothersZiemer}, lead us to conclude:

\smallskip\par
\begin{theorem}\label{minimizersDQOP}
Under the assumptions outlined above, the minimum energy steady states for (DQOP) are monotonically decreasing with respect to distance from the origin, modulo possible translations  within $\mathcal{B}_{R_0}(\delta)$.
\end{theorem}

\smallskip\par
\begin{remark} \label{symmetry}
Since the energy $E(t)$, as well as the problem formulated in (DQOP), are invariant under the transformation $u \rightarrow -u$, if
we  set $u$ to equal $+1$ rather than $-1$  in the $\delta-$annular neighborhood of $\partial \Omega$, then the conclusion in Theorem \ref{minimizersDQOP} would
have yielded that the minimum energy steady states are monotonically decreasing, modulo translation within $\mathcal{B}_{R_0}(\delta)$. Without the constraint that $u$ equals $\pm 1$ in a $\delta-$annular neighborhood of $\partial \Omega$,  the energetics of possible additional steady states would need to be considered. Such additional steady states would include certain energy minimizing  steady state solutions with ``droplet like'' $\pm 1$ concentrations
along the boundary, with lower energy than the axi-symmetric energy minimizing steady states with the same mean mass and their translates, discussed above. For simplicity, we focus here on a more limited set of steady states, which
provide insight into the more general case.
\end{remark}

\smallskip\par
It follows from Theorem \ref{minimizersDQOP} and Remark \ref{symmetry} that we should consider the set of axi-symmetric monotonically increasing steady states for (DQOP)
and their translates that lie within $\mathcal{B}_{R_0}(\delta)$. Since we are looking for constrained mean mass minimizers, we should explore the set of
the monotonically decreasing axi-symmetry steady states with prescribed mean mass, $\bar{u}.$  This is undertaken in detail in the two subsections that follow. If a steady state $u \in \mathcal{K}$ equals
$-1$ in a $\delta-$annular neighborhood of $\partial \Omega$ and increases (non-decreases) monotonically, then
either $(i)$ $u=1$ is attained in a circular neighborhood of the origin or $(ii)$ $u \in [-1,1)$ in $\Omega$ except perhaps at the origin.
In case $(i)$, the steady states contain unique monotonically decreasing annular transition region, see Fig.~\ref{fig1}a and
 Section \ref{subsecAR} for details. In case $(ii)$,  ``dimple solutions'' are possible,  with $-1< u \le 1$ at the origin and with $-1 \le u < 1$ elsewhere; this possibility is explored
 in Section \ref{subsecDR}. See Fig.~\ref{fig1}b.

\smallskip
Before exploring the details of the radial solutions, let us recall that we wish to connect the solutions of (DQOP) with solutions of (SD). In considering the Gamma limit,
the set of (DQOP) solutions are compared with (SD) solutions with similar mass. While solutions to (DQOP) depend on the parameter $\epsilon$, the mean mass constraint is independent of $\epsilon$.
Let us recall (\ref{eqnAu}). If $\bar{u}_0 =-1$, then trivially, $u \equiv -1$. If $\bar{u}_0=+1$, then $u \equiv +1$ is implied by (\ref{eqnAu}), which does not yield a possible solution
due to the requirement that $u=-1$ in a $\delta-$annular neighborhood of $\partial \Omega$. If $\bar{u}_0 \in (-1,1)$ and $\bar{u}_0$ is not too close to $-1$, then  the monotone axi-symmetric steady state solutions to (DQOP)   can be expected to Gamma converge to $-1$ outside the origin in an annulus with width $R_0-r_0$,
and to converge to $+1$ in the disk centered at  the origin with radius $r_0$, with the circular curve $\Gamma$ with radius $r_0$ constituting an axi-symmetric steady state solution to (SD).
Recalling (\ref{eqnAu}), we get that if $\Omega$ is a disk with radius $R_0,$ then
\begin{equation} \label{meanmass_A}
\bar{u}_0=\frac{1}{|\Omega|}(2 \mathcal{A}_0 -|\Omega|) = \frac{2 r_0^2 - R_0^2}{R_0^2}.
\end{equation}
Based on (\ref{meanmass_A}), for monotone steady state axisymmetric solutions of (DQOP) with $\bar{u} \in (-1,1)$,   an $\epsilon$ independent \textbf{equivalent mean mass condition}
is implied,
namely

\smallskip\par
\begin{proposition} \label{Prop1} \textbf{Equivalent mean mass.}
Let $u$ be a monotonically decreasing axi-symmetric steady state solution to (DQOP) with $\bar{u}=\bar{u}_0\in (-1,1)$ which lies within $\Omega=\mathcal{B}_{R_0}(\delta)$, where
$\mathcal{B}_{R_0}(\delta)$ denotes the disk with radius $R_0$ centered at the origin which has a bounding annular neighborhood with width $\delta$ where $u \equiv -1.$
Then $u$ is radial, $u=u(r)$, and
 \begin{equation} \label{emmc}
\int_0^{R_0} u(r) r \, dr = \frac{1}{2} \bar{u} R_0^2 = \frac{1}{2}(2 r_0^2 -R_0^2),
\end{equation}
for some $0 < r_0 < R_0 - \delta.$
\end{proposition}
\bigskip
\begin{figure}[!ht]
\includegraphics[width=.95\textwidth]{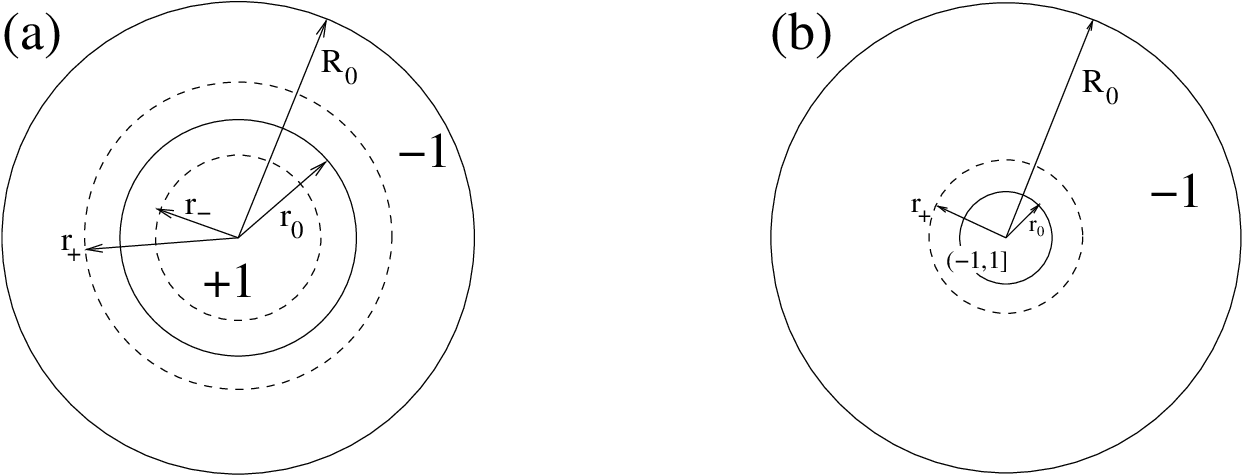}
\caption{a) Schematic portrayal of an annular  solution: $u(r)\equiv 1$ for $0\le r \le r_-$,  $-1 < u(r)<1$ for $r_- < r< r_+$, and $u(r)\equiv -1$ for $r_+ \le r < R_0$.
b) Schematic portrayal of a dimple  solution: $a:=u(0)=\lim_{r \downarrow 0} u(r) \in (-1, 1]$ with $-1 <u(r)\le a$ for $0\le r<r_+$, and $u(r) \equiv -1$ for $r_+ \le r < R_0.$
}\label{fig1}
\end{figure}
\bigskip
\smallskip
Let us consider Fig.~\ref{fig1}. Given $\mathcal{B}_{R_0}(\delta)$, it follows from    (\ref{funcDQOP}) and (\ref{meanmass_A})
for $0< \epsilon \ll 1$, that if $\bar{u}_0$ is not too close to $-1$,
then the solution can be expected to contain an annular transition region, $r_- < r< r_+$, with  $0 < r_- < r_0 < r_+ < R_0-\delta$ and $r_+ - r_- = O(\epsilon)$, where $u(r)=-1$ for $r \in [r_+, R_0]$
and $u(r)=+1$ for $r \in [0, r_-]$; we will refer to such solutions as ``annular'' solutions. If, on the other hand, $\bar{u}_0$ is sufficient close to $-1$, then any possible (admissible) monotone radial solutions will
have $u=-1$ in an annular region, $r_+ \le r \le R_0$, with $0<r_0<r_+<R_0$, such that $u \in (-1, 1)$ for $0<r< r_+$ and $u(0)=\lim_{r \downarrow 0} u(r) \in (-1, 1]$; we will refer to such solutions as ``dimple'' solutions.
In either case, in accordance with (DQOP) and Theorem \ref{oldreg} (see \cite{deF}), $u \in C^1(\mathcal{B}_{R_0}(\delta))$  and outside of the annular or circular regions where $u \equiv \pm 1$,  $u \in C^2$.

Accordingly, the annular solution should satisfy
\begin{equation}\label{annular_r}
\begin{split}
& u+\lambda=-\e^2(u_{rr}+\frac{1}{r}u_r), \qquad  r_-<r<r_+,\\
& u(r_-)=1 \ , \quad u(r_+)=-1,\\
& u_r(r_-)=0\ , \quad u_r(r_+)=0,
\end{split}\end{equation}
for some $r_-$, $r_+$, with $0< r_- < r_0 < r_+< R_0$,
 and for some $\lambda \in R.$ Given the boundary conditions in (\ref{annular_r}), $\lambda$ can be considered here
 as a constant of integration for (\ref{SDi}) with vanishing normal velocity.  By testing the equation in (\ref{annular_r}) by $u(r)$, it is readily seen that $\lambda$ can also
 be viewed as a mean mass conserving Lagrange multiplier for the free energy $E(t)$; see (\ref{lambda_a})  as well as (\ref{lambda_d}) in the sequel.
Note that by (\ref{emmc})
\begin{equation} \label{mean_annulus}
\bar{u}=\frac{(2 r_0^2 - R_0^2)}{R_0^2}=\frac{2}{R_0^2} \int_0^{R_0} u(r)r \,dr,
\end{equation}
and from the prescribed structure of annular solutions
$$\int_0^{R_0} u(r)r\, dr =\int_{r_-}^{r_+} u(r)r\,dr + \frac{1}{2}(r_+^2 + r_-^2 -R_0^2).$$
Using (\ref{annular_r}),
$$\int_{r_-}^{r_+} u(r)r\,dr= -\int_{r_-}^{r_+} [(ru_r)_r - \lambda r] \,dr =\frac{\lambda}{2}(r_+^2 - r_-^2).
$$
Hence the equivalent mean mass condition holds for (\ref{annular_r}) with (\ref{mean_annulus}) if and only if
\begin{equation} \label{lambda_a}
\lambda = \frac{r_+^2 + r_-^2 - 2 r_0^2}{r_+^2 -r_-^2}.
\end{equation}

Similarly, the dimple solution should satisfy
\begin{equation}\label{dimple_r}
\begin{split}
& u+\lambda=-\e^2(u_{rr}+\frac{1}{r}u_r), \qquad  0<r<r_+,\\
& -1<u(0) \le 1 \ , \quad u(r_+)=-1,\\
& u_r(0)=0\ , \quad u_r(r_+)=0,
\end{split}\end{equation}
for some $r_+$, with $0 < r_0<r_+<R_0,$ and for some $\lambda \in R$.  Again,
 by (\ref{emmc}) in Proposition~\ref{Prop1}
\begin{equation} \label{mean_dimple}
\bar{u}=\frac{(2 r_0^2 - R_0^2)}{R_0^2}=\frac{2}{R_0^2} \int_0^{R_0} u(r)r \,dr,
\end{equation}
and from the prescribed structure of dimple solutions
$$\int_0^{R_0} u(r)r\, dr =\int_{0}^{r_+} u(r)r\,dr + \frac{1}{2}(r_+^2  -R_0^2).$$
 Using (\ref{dimple_r}),
$$\int_{0}^{r_+} u(r)r\,dr= -\int_{0}^{r_+} [(ru_r)_r - \lambda r] \,dr =\frac{\lambda}{2}\,r_+^2.
$$
Hence the equivalent mean mass condition holds for  (\ref{dimple_r}) with (\ref{mean_dimple}) if and only if
\begin{equation} \label{lambda_d}
\lambda = \frac{r_+^2  - 2 r_0^2}{r_+^2}.
\end{equation}

\smallskip
Before going into the details of the annular and dimple (DQOP) solutions in the next subsections, we pause to point out that it is possible to attain a large class of additional
solutions by appropriately pasting together translates of dimple and annular solutions of various sizes, so long as the resulting construction lies in $\mathcal{K}$, for some  $\bar{u} \in (-1,1).$
See, for example, the concentrically ringed solutions identified by X.~Chen \cite{XChen} as the asymptotic  limit of solutions to the Cahn-Hilliard equation.

\subsection{Annular solutions for (DQOP)} \label{subsecAR}

We now consider radial ``annular''  monotonically decreasing solutions in $\Omega=\mathcal{B}_{R_0}(\delta)$, containing a transition between the values $\pm 1$, see Fig.~\ref{fig1}a.
More specifically we assume that  $u(r)\equiv +1$ for $r\in [0, r_-],$  $-1<u(r) < 1$ for $r \in (r_-, r_+)$, and $u(r) \equiv -1$ for $r\in [r_+, R_0),$
where $r_\pm$ reflect the location of free boundaries with $0< r_-< r_0 <r_+ < R_0 - \delta.$
 It is possible to seek non-monotone solutions, but
these would have more energy and here we are looking for energy minimizing steady states.
 Generically we may assume that $r_0=O(1)$, although, as we shall see, certain values of $r_0$ with $r_0 = O(\epsilon)$
are also possible. We shall discuss the generic case first, and treat the general case afterwards.

 As we saw  in Section \ref{subsect2.2}, $u(r)$ should satisfy
\begin{equation} \nonumber
\begin{split}
& u+\lambda=-\e^2(u_{rr}+\frac{1}{r}u_r), \qquad  r_- <r<r_+,\\
& u(r_-)= 1, \quad u(r_+)=-1,\\
& u_r(r_-)=0, \quad u_r(r_+)=0,
\end{split}\end{equation}
where  the equivalent mean mass condition holds if and only if
\begin{equation} \label{annular_mm}
\bar{u}=\frac{(2 r_0^2 - R_0^2)}{R_0^2}, \quad
\lambda = \frac{r_+^2 + r_-^2  - 2 r_0^2}{r_+^2 - r_-^2}, \quad r_0=\Bigl[\frac{1+ \bar{u}}{2}\Bigr]^{1/2} R_0,
\end{equation}
for some $0 < r_0 <r_+$, see (\ref{annular_r})--(\ref{lambda_a}).
Thus $\lambda=\lambda(r_-, r_+, r_0)$ and $r_0=r_0(\bar{u}, R_0).$  We shall see that annular solutions exist for  $-1 + O(\epsilon^2) < \bar{u} < 1 -4\delta/R_0 + O(\delta^2, \epsilon)$ with
$O(\epsilon) < r_0 < R_0 - \delta + O(\epsilon);$  some further specifics to follow.

The parameters, $r_\pm$, reflecting the location of the free boundaries, are to be determined. Thus, we have four boundary conditions,
as well as the monotonicity and range constraints.
In solving (\ref{annular_r})--(\ref{lambda_a}), we have four degrees of freedom, two from  the second order ODE in (\ref{annular_r}) and two from  the parameters  $r_\pm,$ with  $\lambda=\lambda(r_-, r_+, r_0)$. Hence, one would expect
the possible annular solutions to be uniquely determined by
 $\bar{u}$ (or, equivalently, by $r_0$ and $R_0$).

Setting
\begin{equation}\label{rescaling_a}
q=\frac{r}{\e},  \ \quad
     q_0=\frac{r_0}{\e}, \quad
    q_\pm =\frac{r_\pm}{\e}, \quad Q_0=\frac{R_0}{\e}, \quad \hbox{and} \quad v(q):=[u(r)+\lambda]|_{r=\epsilon q},
\end{equation}
    we get the following problem for $v(q)$,
\begin{equation}\label{annular_v}
\begin{split}
& qv_{qq}+v_q+qv=0, \quad 0<\qm<q<\qp,\\
& v(\qm)=1+\lambda,  \quad  v(\qp)=-1+\lambda,\\
& v_q(\qm)=0, \quad v_q(\qp)=0,
\end{split}\end{equation}
with
\begin{equation} \label{lambda_v}
\lambda=\frac{q_+^2 + q_-^2 - 2 q_0^2}{q_+^2 - q_-^2}.
\end{equation}
The equation in (\ref{annular_v}) is a Bessel equation of order zero, \cite[10.2.1]{dlmf},
whose general solution (\cite[10.2(iii),10.6.3]{dlmf}) is
\begin{equation}\label{BesselJY}
v(q)=c_1 J_0(q)+c_2 Y_0(q), \quad q>0, \quad\quad c_1, c_2 \in R,
\end{equation}
with
\begin{equation}\nonumber
v_q(q)=-c_1 J_1(q)-c_2 Y_1(q), \quad q >0.
\end{equation}
The coefficients $c_1,$ $c_2$, in the solution to (\ref{annular_v}), should depend on $q_\pm$ and $\lambda,$  which in turn depend on $\bar{u}$ (or equivalently on $r_0$), as well as on the underlying parameters $R_0$ and $\e$.


\smallskip
Using the polar representation for Bessel functions, \cite[10.18.4, 10.18.6, 10.18.7]{dlmf},
we get that
\begin{align}
& v(q)=c_1 J_0(q)+c_2 Y_0(q) =
   M_0(q)\big( c_1 \cos \theta_0(q) + c_2 \sin \theta_0(q) \big), \quad q > 0, \label{v1}\\
& -v_q(q)=c_1 J_1(q)+c_2 Y_1(q) =
   M_1(q)\big( c_1 \cos \theta_1(q) + c_2 \sin \theta_1(q) \big), \quad q >0, \label{vq1}
\end{align}
where for $n=0,1,$  $\theta_n(x):=\arctan(Y_n(x)/J_n(x))$ for $x>0,$ and
$$
M_n(x):=\sqrt{J_n^2(x)+Y_n^2(x)}>0, \quad x>0,
$$
since the Bessel functions $J_n,Y_n$ do not vanish simultaneously.

Setting
\begin{equation} \label{c1c2A}
\cos\p=\frac{c_1}{\sqrt{c_1^2+c_2^2}}, \quad
\sin\p=\frac{c_2}{\sqrt{c_1^2+c_2^2}}, \quad \hbox{and}\quad A:=\frac{1}{\sqrt{c_1^2+c_2^2}},
\end{equation}
we can write (\ref{v1}),(\ref{vq1}) as
\begin{equation}\label{AMcq}
v(q)= A M_0(q) \cos \big(\theta_0(q) - \p \big), \quad
v_q(q)= -A M_1(q) \cos \big(\theta_1(q) - \p \big), \quad q > 0.
\end{equation}

Since $M_1$ and $A$ are positive, the boundary conditions  $v_q(\qm)=v_q(\qp)=0$ in (\ref{annular_v}) imply that
\begin{equation}
 \cos \big(\theta_1(\qm) - \p \big)=\cos \big(\theta_1(\qp) - \p \big)=0.
\end{equation}
As we are seeking monotone solutions, $v(q)$ should be monotonically decreasing  with $v_q<0$ for $q_-<q<q_+$ by Sturmian theory. The positivity of $M_1$ and $A$ now
implies that  $\cos \big(\theta_1(q) - \p \big)>0$ for $q_-<q< q_+$, and thus that
\begin{equation} \label{arange_a}
\frac{-\pi}{2}  < \theta_1(q) - \p < \frac{\pi}{2}, \quad
                  0<\qm<q<\qp,
\end{equation}
up to possible translations by $2 k \pi$,  $k \in Z,$ which do not effect the solutions. The function $\theta_1(q)$ is monotonically increasing (\cite[10.18.18]{dlmf});
therefore
\begin{equation}\label{theta1_phi}
\theta_1(\qm)-\p=-\frac{\pi}{2}\ , \quad
\theta_1(\qp)-\p=\frac{\pi}{2}, \quad \hbox{and}\quad \theta_1(\qp)-\theta_1(\qm)=\pi.
\end{equation}
Since $\lim_{q \downarrow 0} \theta_1(q)= - \frac{\pi}{2}$ (\cite[10.18.3]{dlmf}),  the monotonicity of $\theta_1(q)$ and
(\ref{theta1_phi}) imply that
\begin{equation} \label{theta1_phi2}
\theta_1(q_+)> \frac{\pi}{2}, \quad \theta_1(q_-) > - \frac{\pi}{2}, \quad \p >0.
\end{equation}

\subsubsection{Annular solutions, the generic case}\label{subsecGC}

We remarked earlier that  $-1 + O(\epsilon^2) < \bar{u} < 1  -4\delta/R_0 + O(\delta^2, \epsilon)$ with
$O(\epsilon) < r_0 < R_0 - \delta + O(\epsilon)$. Thus,  generically $r_0=O(1)$.  Considerations of scaling and energy minimization of the energy
prescribed in (\ref{funcDQOP}) for $0 <\epsilon
 \ll 1$  imply that transition widths between the phases for energy minimizing solutions scale as $O(\e)$, see e.g. \cite{NC_Meccanica,PSternberg88},
and hence, $r_+ - r_- = O(\e)$ and $r_-$, $r_+$ are $O(1)$ when $r_0 = O(1)$. This will also be demonstrated directly in Section \ref{EandU}.
Accordingly, the rescalings in (\ref{rescaling_a})  imply that
\begin{equation} \label{generic_q}
q_-, \, q_0,\, q_+ \, = \,O(\e^{-1}), \quad q_+ - q_- =O(1), \quad \epsilon \, q_0 = O(1)
\end{equation}
in the generic case. Throughout this subsection, we  assume that (\ref{generic_q}) holds.

\bigskip
For large values of $x$ (see \cite[10.18.18]{dlmf}),
\begin{equation}
\theta_1(x)=x-\frac{3\pi}{4}+\frac{3}{8x}+O\left(\frac{1}{x^3}\right), \quad x \gg 1.
\end{equation}
Hence, by (\ref{theta1_phi}), for large values of $\qm,$ $\qp,$
\begin{equation}\label{q_large_l}
\qm -\frac{3\pi}{4}+\frac{3}{8 \qm}+O\left(\frac{1}{q_0^3}\right) - \p=
 - \frac{\pi}{2},
\quad
\qp-\frac{3\pi}{4}+\frac{3}{8\qp}+O\left(\frac{1}{q_0^3}\right) - \p=
 \frac{\pi}{2},
 \end{equation}
and thus
\begin{equation}\label{eq:eq38}
\qp-\qm = \pi + \frac{3}{8}\left(\frac{1}{\qm}-\frac{1}{\qp}\right) +
                 O\left(\frac{1}{q_0^3}\right).
\end{equation}

The boundary conditions $v(\qm)= 1+\lambda$ and $v(\qp)=-1 + \lambda$ in (\ref{annular_v}) imply that
\begin{equation}\label{eq:eq39}
  A M_0(\qm) \cos \big(\theta_0(\qm) - \p \big)=1+\lambda, \quad
A M_0(\qp) \cos \big(\theta_0(\qp) - \p \big)=-1+\lambda.
\end{equation}
Subtracting the equations in (\ref{eq:eq39}),
\begin{equation}\label{eq:eq40}
M_0(\qm) \cos \big(\theta_0(\qm) - \p \big) -
M_0(\qp) \cos \big(\theta_0(\qp) - \p \big)=\frac{2}{A}.
\end{equation}
For large values of $x$ (see \cite[10.18.18]{dlmf}),
\begin{equation}\nonumber
\theta_0(x)=x-\frac{\pi}{4}-\frac{1}{8x}+O\left(\frac{1}{x^3}\right), \quad x \gg 1.
\end{equation}
Hence using (\ref{q_large_l})
\begin{equation}\nonumber
\theta_0(\qm)-\p =
           -\frac{1}{2\qm}+O\left(\frac{1}{q_0^3}\right), \quad \theta_0(\qp)-\p ={\pi} -\frac{1}{2\qp} +O\left(\frac{1}{q_0^3}\right),
           \end{equation}
           and therefore
\begin{equation} \label{qlarge32}
\cos(\theta_0(\qm)-\p) =1- \frac{1}{8\qm^2}
                                     +O\left(\frac{1}{q_0^3}\right), \quad
\cos(\theta_0(\qp)-\p)
    =-1+\frac{1}{8\qp^2}+O\left(\frac{1}{q_0^3}\right).
\end{equation}
For large values of $x$ (see \cite[10.18.17]{dlmf}),
\begin{equation}\label{q_large_33}
M_0(x)=\sqrt{\frac{2}{\pi x}}+O\left(\frac{1}{x^{5/2}}\right), \quad x \gg 1.
\end{equation}
Using (\ref{qlarge32}), (\ref{q_large_33}) in (\ref{eq:eq39}),
\begin{equation} \nonumber
\sqrt{\frac{2}{\pi \qm}} \cdot \left( 1- \frac{1}{8\qm^2} \right) -
\sqrt{\frac{2}{\pi \qp}} \cdot \left( -1+\frac{1}{8\qp^2} \right)
             + O\left(\frac{1}{q_0^{5/2}}\right) = \frac{2}{A},
\end{equation}
which implies that
\begin{equation}\label{q_large_A}
A=\frac{\sqrt{2\pi}}{\frac{1}{\sqrt{\qm}} + \frac{1}{\sqrt{\qp}}}
   +O\left(\frac{1}{q_0^{3/2}}\right).
\end{equation}

Returning  to \eqref{eq:eq39} and  summing the two equations,
\begin{equation}\nonumber
A M_0(\qm) \cos \big(\theta_0(\qm) - \p \big) +
A M_0(\qp) \cos \big(\theta_0(\qp) - \p \big)=2\lambda,
\end{equation}
and then using the approximations in (\ref{qlarge32}), (\ref{q_large_33}),
\begin{equation}\nonumber
\left(\frac{\sqrt{2\pi}}{\frac{1}{\sqrt{\qm}} + \frac{1}{\sqrt{\qp}}}
   +O\left(\frac{1}{q_0^{3/2}}\right)\right)
\left( \sqrt{\frac{2}{\pi \qm}} - \sqrt{\frac{2}{\pi \qp}}
            + O\left(\frac{1}{q_0^{5/2}}\right) \right) = 2\lambda,
\end{equation}
which implies that
\begin{equation} \label{q_large_lapprox}
\lambda = \frac{\sqrt{\qp}-\sqrt{\qm}}{\sqrt{\qp}+\sqrt{\qm}}
                               +O\left(\frac{1}{q_0^2}\right).
\end{equation}
Using now \eqref{eq:eq38},
\begin{equation} \nonumber
\lambda =
   \frac{\pi+O\left(\frac{1}{q_0}\right)}{(\sqrt{\qm+\pi}+\sqrt{\qm})^2}
= O\left(\frac{1}{q_0}\right),
\end{equation}
%
from which we get that
\begin{equation} \label{q_large_42}
\qm = \frac{\pi(1-2\lambda)}{4\lambda}+O\left(\frac{1}{q_0}\right), \quad
\qp = \frac{\pi(1+2\lambda)}{4\lambda}+O\left(\frac{1}{q_0}\right),
\end{equation}
which implies that $q_\pm=O(q_0)$ and $q_+  -  q_-=O(1)$, in accordance with (\ref{generic_q}).

Using (\ref{q_large_42}) and (\ref{rescaling_a}) in the expression for $\lambda$ given in (\ref{annular_mm})
and noting that (\ref{q_large_lapprox}) implies that $\lambda >0$, then
solving for $\lambda$, we obtain that
\begin{equation} \label{q_large_lambdaf}
\lambda=  \frac{\pi}{4q_0}+O\left(\frac{1}{q_0^2}\right).
\end{equation}
%
Substituting the above expression for $\lambda$ into (\ref{q_large_42}) we get that
\begin{equation}\nonumber
\qm = q_0-\frac{\pi}{2}+O\left(\frac{1}{q_0}\right) \quad
\qp = q_0+\frac{\pi}{2}+O\left(\frac{1}{q_0}\right).
\end{equation}

In order to obtain an approximate solution $v(q)$ to (\ref{annular_v}), and subsequently to obtain an approximate
solution $u(r)=v(r/\epsilon)-\lambda$ to (\ref{annular_r}), based, say on the expression in (\ref{BesselJY}), in the generic case, it remains
to identify approximations for the coefficients, $c_1$, $c_2$.

The boundary conditions $v(q_+)=-1 + \lambda$ and $v_q(q_+)=0$ in (\ref{annular_v}),
imply that
\begin{equation} \label{c1c2}
c_1=\frac{(-1 + \lambda) Y_1(q_+)}{J_0(q_+) Y_1(q_+) + J_1(q_+) Y_0(q_+)}, \quad c_2=\frac{(-1 + \lambda) Y_0(q_+)}{J_0(q_+) Y_1(q_+) + J_1(q_+) Y_0(q_+)}.
\end{equation}

%
%
%

We know  (see \cite[10.17.3, 10.17.4]{dlmf}) that
\begin{equation}\nonumber
\begin{split}
& J_{0}(x)=\sqrt{\frac{2}{\pi x}}
        \left(\cos\left(x-\frac{\pi}{4}\right)
           +O\left(\frac{1}{x}\right)\right), \quad  x\gg 1, \\
& Y_{0}(x)=\sqrt{\frac{2}{\pi x}}
        \left(\sin\left(x-\frac{\pi}{4}\right)
           +O\left(\frac{1}{x}\right)\right), \quad  x\gg 1,
\end{split}\end{equation}
and that
\begin{equation}\nonumber
\begin{split}
& J_{1}(x)=\sqrt{\frac{2}{\pi x}}
        \left(\cos\left(x-\frac{3\pi}{4}\right)
           +O\left(\frac{1}{x}\right)\right), \quad  x \gg 1, \\
& Y_{1}(x)=\sqrt{\frac{2}{\pi x}}
        \left(\sin\left(x-\frac{3\pi}{4}\right)
           +O\left(\frac{1}{x}\right)\right), \quad  x\gg 1.
\end{split}\end{equation}
Using the approximations above in (\ref{c1c2}),
\begin{equation}
c_1 =  -\sqrt{\frac{q_0 \pi}{2}}\, \frac{[\sin(q_0- \pi/4)]}{ \sin(2 q_0)} + O\left(\frac{1}{q_0}\right),\quad c_2 =  -\sqrt{\frac{q_0 \pi}{2}}\, \frac{ [\cos(q_0 - \pi/4)]}{ \sin(2 q_0)} + O\left(\frac{1}{q_0}\right),
\end{equation}
and then returning to (\ref{BesselJY}), we obtain  that
\begin{equation} \label{v_generic}
v(q)=  - \sqrt{\frac{q_0}{q}}\, \frac{\cos(q+q_0)}{\sin(2 q_0)} + O\left(\frac{1}{q_0}\right), \quad q_- < q < q_+ \; ( q_0-\frac{\pi}{2}    < q < q_0 + \frac{\pi}{2} \;) .
\end{equation}
Recalling   (\ref{rescaling_a}),
\begin{equation} \label{u_generic}
u(r) = - \sqrt{\frac{r_0}{r}} \, \frac{\cos((r+r_0)/\epsilon)}{\sin(2 r_0/{\epsilon})}  + O\left(\frac{\epsilon}{r_0}\right), \quad r_- < r <r_+.
\end{equation}

\subsubsection{Existence and Uniqueness}\label{EandU}

In this section, we prove
\smallskip
\begin{theorem} \label{existence_r}
Given $\epsilon$, $\delta$, and $R_0$,  $0<\epsilon \ll \delta \ll 1,$  $R_0=O(1)$.
There exists a unique solution to (\ref{annular_r}), (\ref{lambda_a}) in $\mathcal{B}_{R_0}(\delta)$
for $r_0 \in (r_0^{\inf}, \, r_0^{\sup}]$, where
$r_0^{\inf}=\bar{q}/\epsilon$ with $\bar{q}:=\theta_1^{-1}(\pi/2)$ and $r_0^{\sup}=R_0-\delta +O(\epsilon).$
\end{theorem}

\begin{remark}
In Theorem \ref{existence_r}, $r_0^{\inf}$ corresponds to the largest value of $r_0$  for which dimple solutions exist, to be discussed in detail in Section \ref{subsecDR};
 for the annular solutions, $r_- \downarrow 0$ as $r_0 \downarrow r_0^{\inf}$. The upper limit,  $r_0^{\sup}$, corresponds to the largest possible value of $r_0$, given the  $\delta-$width annulus
 where $u\equiv -1$ within $\mathcal{B}_{R_0}(\delta)$, with $r_+ \uparrow R_0 - \delta$ as $r_0 \uparrow r_0^{\sup}$.
\end{remark}

\begin{proof}
It is convenient to  prove Theorem \ref{existence_r} by using the rescalings defined in (\ref{rescaling_a}), and proving unique existence for the equivalent problem
prescribed in (\ref{annular_v}), (\ref{lambda_v}) for $q_0 \in (q_0^{\inf}, q_0^{\sup}]$, where $q_0^{\inf}=\bar{q}=\epsilon r_0^{\inf}$, $q_0^{\sup} =\epsilon r_0^{\sup}$.
In Proposition \ref{existu_q} below,  unique existence is demonstrated for   $q_0 \in (\bar{q}, \infty)$, and then the implied ranges indicated in Theorem \ref{existence_r} follows by
returning to the original scaling and
imposing $r^{\sup}_0$  the upper limit, with $r^{\sup}_0$  corresponding to $r^{\sup}_+=R_0 -\delta$. In the course of the proof of  Proposition \ref{existu_q}, we shall see that
$r_0^{\sup} -r_+^{\sup} =O(\epsilon)$, which implies that  $r_0^{\sup}=R_0-\delta +O(\epsilon)$.

\begin{proposition} \label{existu_q}
There exists a unique monotone solution to (\ref{annular_v})-(\ref{lambda_v}) for every $q_0 \in (\bar{q}, \infty)$, where
$\bar{q}:=\theta_1^{-1}(\pi/2)$ corresponds to the first positive root of $J_1(q)$.
\end{proposition}

\smallskip
\begin{proof}
The proof relies on introducing a functional, $D=D(t)$,  defined below, which allows us to focus first on existence and then on uniqueness.
Let us recall that the general solution to the equation in (\ref{annular_v})
may be written as in (\ref{AMcq}), namely as
\begin{equation}\nonumber
v(q)= A M_0(q) \cos \big(\theta_0(q) - \p \big), \quad
v_q(q)= -A M_1(q) \cos \big(\theta_1(q) - \p \big), \quad q \ge 0,
\end{equation}
where $A>0$, $M_0(q)$, $M_1(q)>0,$ for $q \ge 0.$ The boundary conditions  $v_q(\qm)=v_q(\qp)=0$ in (\ref{annular_v}), together with  the monotonicity of $\theta_1(x)$
with $\lim_{x \downarrow 0} \theta_1(x)= -\pi/2$, imply that
\begin{equation} \label{arange_aa}
   \p >0, \quad            \frac{-\pi}{2}  < \theta_1(q) - \p < \frac{\pi}{2}, \quad  0<\qm<q<\qp,  \quad \theta_1(q_\pm)-\p= \pm \frac{\pi}{2},
\end{equation}
up to possible translations by $2 k \pi$,  $k \in Z,$ see (\ref{arange_a})--(\ref{theta1_phi2}).

\bigskip
The general solution contains two parameters, $A$ and $\p$. As we have already accommodated the boundary conditions $v_q(\qm)=v_q(\qp)=0$,
 finding a solution to (\ref{annular_v}), (\ref{lambda_v}) entails identifying $q_-$ and $q_+$ as functions of $A$ and $\p$ so as to satisfy the two remaining boundary conditions in (\ref{annular_v}).
  In view of (\ref{arange_aa}),  it is convenient to work with the parameters $A$ and $t$, rather than with the parameters $A$ and $\p$, where
 \begin{equation}\label{deft}
 t=\theta_1(q_0) - \p, \quad -\frac{\pi}{2} \le t \le \frac{\pi}{2},
\end{equation}
 up to possible translation by $2k\pi$, $k \in Z$.
From (\ref{arange_aa}), (\ref{deft}), we get  that
\begin{equation} \label{theta1qpm}
\p(t)=\theta_1(q_0)-t,  \quad \theta_1(q_\pm(t))=\theta_1(q_0)-t \pm \frac{\pi}{2}, \quad  t \in [-\pi/2, \pi/2].
\end{equation}
Since  $\theta_1(x)$ is monotonically increasing and continuously differentiable for $x>0$ (\cite[10.18.8]{dlmf}),  $\theta_1(x)$ has an inverse function which is also monotonically increasing and continuously differentiable, and
\begin{equation}\label{46}
q_\pm(t)=\theta_1^{-1}\left(\theta_1(q_0)-t \pm \frac{\pi}{2}\right).\end{equation}
Thus $\p$, $\qm$, $\qp$, and well as $\lambda$ (see (\ref{lambda_v})), can be viewed as  continuously differentiable functions of $t$,
since $\theta_1(q_0) -t -\pi/2 > -\pi/2$ for $t \in [-\pi/2, \pi/2]$, for
$q_0 > \theta_1^{-1}(\pi/2)=\bar{q}.$

From (\ref{AMcq}) and the  boundary conditions, $v(q_-)=1+ \lambda$ and $v(q_+)=-1 +\lambda$,
\begin{equation} \label{A1A2}
     A M_0(\qm) \cos \big(\theta_0(\qm) - \p \big)=1+\lambda, \quad
A M_0(\qp) \cos \big(\theta_0(\qp) - \p \big)=-1+\lambda.
\end{equation}
As $\p$, $q_-$, $q_+$, and $\lambda$ are prescribed in terms of $t$, we get two descriptions for $A$ in terms of $t$ from (\ref{A1A2})  as long as $\cos \big(\theta_0(\qm) - \p \big)$ and
$\cos \big(\theta_0(\qp) - \p \big)$ do not vanish, since  $M_0(x)>0$  for $x>0$. Moreover,  the two descriptions must be equal  for some value of $t \in [-\pi/2, \pi/2]$, if  we are to attain a solution to (\ref{annular_v}), (\ref{lambda_v}).
\smallskip
First we prove
\smallskip
\begin{lemma} \label{signcosine}
 $\cos \big(\theta_0(\qm) - \p \big)>0,$ and $\cos \big(\theta_0(\qp) - \p \big)<0.$
\end{lemma}

\begin{proof}
First we note that for $x \ge 0,$
\begin{equation}\label{MN}
\begin{split}
& J_0(x)=M_0(x)\cos\theta_0(x), \quad  J_0'(x)=N_0(x)\cos\phi_0(x) =  -M_1(x)\cos\theta_1(x),\\
& Y_0(x)=M_0(x)\sin\theta_0(x), \quad  Y_0'(x)=N_0(x)\sin\phi_0(x)  = -M_1(x)\sin\theta_1(x).
\end{split}\end{equation}
Using (\ref{MN}), the formula \cite[10.18.12]{dlmf},
$M_0(x) N_0(x) \sin(\phi_0(x)-\theta_0(x))=\frac{2}{\pi x}, \quad x>0, $
and a little trigonometry,
\begin{equation} \label{MM}
M_0(x) M_1(x) \sin(\theta_0(x)-\theta_1(x))=\frac{2}{\pi x}, \quad x>0.
\end{equation}
Since $M_0(x),\, M_1(x) >0$ for $x>0$, (\ref{MM}) implies that $\sin(\theta_0(x)-\theta_1(x))>0, \quad x>0$. Thus
\begin{equation}\nonumber
  0<\theta_0(x)-\theta_1(x)<\pi, \quad x>0,
\end{equation}
up to  possible translation by $2k\pi$, $k \in Z$. Therefore,
\begin{equation} \label{t1t0}
\theta_1(x)<\theta_0(x)<\theta_1(x)+\pi, \quad  x>0.
\end{equation}
Setting $x=\qm$ in (\ref{t1t0}), we obtain that $\theta_1(\qm)-\p< \theta_0(\qm)-\p<\theta_1(\qm)-\p+\pi.$ Then using (\ref{arange_aa}),
\begin{equation}\nonumber
 -\frac{\pi}{2}< \theta_0(\qm)-\p<\frac{\pi}{2}
\end{equation}
and therefore $\cos \big(\theta_0(\qm) - \p \big)>0.$
Similarly, we obtain from (\ref{t1t0}) and (\ref{arange_aa}) that
\begin{equation}\nonumber
 \frac{\pi}{2}< \theta_0(\qp)-\p<\frac{3\pi}{2}
\end{equation}
and therefore $\cos \big(\theta_0(\qp) - \p \big)<0.$
\end{proof}

Given Lemma \ref{signcosine}, we now obtain two expressions for $A$, namely,
\begin{equation} \label{defA}
A=\frac{1+\lambda}{M_0(\qm) \cos \big(\theta_0(\qm) - \p \big)}=
  \frac{-1+\lambda}{M_0(\qp) \cos \big(\theta_0(\qp) - \p \big)}.
\end{equation}
Using the definition of $\lambda$ given in (\ref{lambda_a}), and  noting that  $\theta_1(\qp)-\theta_1(\qm)=\pi$     implies that $q_+ - q_->0,$
we obtain from (\ref{defA}) that
%
\begin{equation} \nonumber
 \frac{\qp^2-q_0^2}{M_0(\qm) \cos \big(\theta_0(\qm) - \p \big)} = \frac{\qm^2-q_0^2}{M_0(\qp) \cos \big(\theta_0(\qp) - \p \big)},
 \end{equation}
where all the terms  are continuously differentiable functions of $t$. It remains to verify that the above equation uniquely defines $t$.
Cross-multiplying, we obtain now that roots of $D(t)=0,$ with $t \in [-\pi/2, \pi/2],$ correspond to solutions of (\ref{annular_v}), (\ref{lambda_v}), where
\begin{equation}\label{defD}
D(t) := \mathcal{D}(\qp(t)) - \mathcal{D}(\qm(t)),\quad \mathcal{D}(q):=(q^2-q_0^2)M_0(q) \cos \big(\theta_0(q) - \p \big),
\end{equation}
is continuously differentiable for $t \in [-\pi/2, \pi/2]$.
Existence now follows from  the following lemma.
\smallskip
\begin{lemma} \label{Dpm}
$D\left(-{\pi}/{2}\right)<0, \quad
   D\left({\pi}/{2}\right)>0.$
\end{lemma}

\begin{proof}   Let us first consider $D(t)$ with $t=-{\pi}/{2}.$ From (\ref{theta1qpm}),  $\theta_1(q_-(-\pi/2))=\theta_1(q_0)>\pi/2$
for $q_0> \bar{q}$.
Recalling the monotonicity of  $\theta_1(x)$ for $x>0$, and using (\ref{46}),
\begin{equation} \label{tm}
\qm(-\pi/2) =q_0, \quad  \qp(-\pi/2) =
    \theta_1^{-1}\left(\theta_1(q_0)+\pi\right)> q_0.
\end{equation}
Similarly for $D(t)$ with $t={\pi}/{2},$ we obtain that $\theta_1(q_+(\pi/2))=\theta_1(q_0)> \pi/2,$ and hence
\begin{equation}\label{tp}
\qm(\pi/2) = \theta_1^{-1}\left(\theta_1(q_0) - \pi\right)< q_0, \quad \qp(\pi/2) =q_0.
\end{equation}

 \smallskip
Let us recall Lemma \ref{signcosine} and the positivity of $M_0(x)$ for $x>0$. Then using (\ref{tm}) and the definition  of $D(t)$ in (\ref{defD}),
\begin{equation}\nonumber
D\left(-{\pi}/{2}\right)=(\qp(-\pi/2)^2-q_0^2)M_0(\qp(-\pi/2)) \cos \big(\theta_0(\qp(-\pi/2)) - \p(-\pi/2) \big) < 0.
\end{equation}
Similarly, using (\ref{tp}) in (\ref{defD}), we obtain that $D\left({\pi}/{2}\right)>0.$
\end{proof}


Uniqueness now follows by proving the following
\begin{lemma} \label{Dmonotone} $D'(t) >0,$  $t\in (-\pi/2, \pi/2).$
\end{lemma}

\begin{proof}
Our proof  is based on calculating $D'(t)$ where  $D(t)$, given in (\ref{defD}).
We begin by calculating $\p'(t)$, $q_\pm'(t)$, $\theta_0'(q_\pm(t))q_\pm'(t)$, and     $M_0'(q_\pm(t))q_\pm'(t).$

\smallskip
We recall the formulas from \cite[10.18.8]{dlmf},
\begin{equation} \label{Mt01d}
M_0^2(x)\theta_0'(x)=\frac{2}{\pi x}, \quad  M_1^2(x)\theta_1'(x)=\frac{2}{\pi x}, \quad x > 0.
\end{equation}
 Since $M_0(x)$ and $M_1(x)>0$ are positive, (\ref{Mt01d}) implies that
\begin{equation} \label{dtheta01}
\theta_0'(x)=\frac{2}{\pi x}\frac{1}{\big(M_0(x)\big)^2}, \quad
\theta_1'(x)=\frac{2}{\pi x}\frac{1}{\big(M_1(x)\big)^2}\quad x>0.
\end{equation}
It follows from  (\ref{deft}) that
\begin{equation} \label{dpsi}
\p'(t) = -1.
\end{equation}
From  (\ref{theta1qpm}) and the invertibility of $\theta_1(x)$ for $x>0$,
\begin{equation} \label{qpmt}
q_\pm=q_\pm(t)=\theta_1^{-1}\left(\theta_1(q_0)-t \pm \frac{\pi}{2}\right),
\end{equation}
and therefore  $\theta_1'(q_\pm(t))q_\pm'(t)=-1$.
Recalling  (\ref{dtheta01}), we obtain that
\begin{equation}\label{difq}
q_\pm'(t)=-\frac{\pi q_\pm}{2} \left( M^2_1(q_\pm)\right), \quad \theta_0'(q_\pm(t))q_\pm'(t)=-\left(\frac{M^2_1(q_\pm)}{M^2_0(q_\pm)}\right).
\end{equation}

Let us consider now the formula \cite[10.18.11]{dlmf},
${M_0\theta_0'}/{M_0'}=\tan(\phi_0-\theta_0)$. Recalling (\ref{MN}),
we see that $\tan\phi_0=\tan\theta_1$, and
\begin{equation}\label{62}
\frac{\sin\phi_0}{\cos\phi_0}=\frac{N_0\sin\phi_0}{N_0\cos\phi_0}=
\frac{-M_1\sin\theta_1}{-M_1\cos\theta_1},
\end{equation}
and therefore using (\ref{difq}),
\begin{equation}\nonumber \begin{split}
    &M_0'(q_\pm)=M_0(q_\pm)\theta_0'(q_\pm)\cot(\theta_1(q_\pm)-\theta_0(q_\pm)), \\[1ex]
&M_0'(q_\pm)q_\pm' = \frac{M_1^2(q_\pm)}{M_0(q_\pm)}\cot(\theta_0(q_\pm)-\theta_1(q_\pm)).
\end{split}\end{equation}
Since $\theta_0(\qp)-\theta_1(\qp)=
  \big(\theta_0(\qp)-\p\big)-\big(\theta_1(\qp)-\p\big) =
  \big(\theta_0(\qp)-\p\big)-\frac{\pi}{2}$, we get that
\begin{equation}\label{64}
\cot\big(\theta_0(\qp)-\theta_1(\qp)\big)=-\tan\big(\theta_0(\qp)-\p\big),
\end{equation}
and similarly
\begin{equation}\label{65}
\cot\big(\theta_0(\qm)-\theta_1(\qm)\big)=-\tan\big(\theta_0(\qm)-\p\big).
\end{equation}
Thus
\begin{equation}
M_0'(q_\pm(t))q_\pm'(t)=-\frac{M^2_1(q_\pm)}{M_0(q_\pm)}\tan\big(\theta_0(q_\pm)-\p\big).
\end{equation}

Recalling (\ref{defD}), let us differentiate $\mathcal{D}(q_\pm(t)).$ Using the results above,
\begin{equation}\begin{split}
&\frac{d}{dt}\mathcal{D}(q_\pm(t)=(q_\pm^2-q_0^2)M_0(q_\pm)\cos(\theta_0(q_\pm)-\p) \\
&\quad=-\pi q_\pm^2 M_0(q_\pm) M^2_1(q_\pm) \cos(\theta_0(q_\pm)-\p)
             -(q^2-q_0^2)M_0(q)\sin(\theta_0(q)-\p).
\end{split}\end{equation}

We note that
$ \sin(\theta_0(q)-\p) = \tan(\theta_0(q)-\p)\cos(\theta_0(q)-\p),$
and using (\ref{62})-(\ref{65})) we get that $\tan(\theta_0(q)-\p)=+\cot(\phi_0(q)-\theta_0(q))$.
By formula \cite[10.18.11]{dlmf},
$\tan(\phi_0-\theta_0)={2}/({\pi x M_0 M_0'})$ for $x>0$, and hence
\begin{equation}
\tan(\theta_0(q)-\p)=\cot(\phi_0(q)-\theta_0(q))=\half\pi q M_0M_0'=
\tfrac{1}{4}\pi q \left(M_0^2\right)'.
\end{equation}
Therefore
\begin{equation}
\frac{d}{dt}\mathcal{D}(q_\pm(t))
=-\left(q_\pm  M^2_1(q_\pm)
           +\tfrac{1}{4}(q_\pm^2-q_0^2)\left(M_0^2(q_\pm)\right)'\right)
                   \pi q_\pm   M_0(q_\pm)  \cos(\theta_0(q_\pm)-\p),
\end{equation}
where  $q_\pm=q_\pm(t)$, $\p=\p(t)$.
%
Let us, then, consider the expression
\begin{equation}
P(x):= x \big(M_1(x)\big)^2
     +\tfrac{1}{4}(x^2-q_0^2)\left(M_0(x)^2\right)', \quad x>0.
\end{equation}
Since $N_0^2(x)=M_1^2(x)$ for $x \ge 0$ by  (\ref{MN}),  the
formula \cite[10.18.10]{dlmf},
$x^2 M_0(x)M_0'(x)+x^2 N_0(x)N_0'(x)+x\big(N_0(x)\big)^2=0,$ $x \ge 0$, implies that
%
%
\begin{equation} \nonumber
x\big(M_1(x)\big)^2 = -\half x^2(M_0^2(x))'-\half x^2 (M_1^2(x))', \quad x \ge 0.
\end{equation}
Therefore,
\begin{equation}\nonumber
P(x) = -\tfrac{1}{4} x^2\left(M^2_0(x)\right)'
     -\tfrac{1}{4}q_0^2\left(M^2_0(x)\right)'
      -\tfrac{1}{2} x^2 \left(M^2_1(x)\right)', \quad x > 0.
\end{equation}
The claim below implies that $P(x)>0$ for   $x>0$.
\smallskip

\begin{claim} \label{MMdecrease}
$\left(M^2_0(x)\right)'<0$ and $\left(M^2_1(x)\right)'<0$, for $x >0.$
\end{claim}

\begin{proof}
To prove the claim, we use the Nicholson's Integral  Representation (see \cite[10.9.30]{dlmf}),
which implies that for $x>0,$
\begin{equation}\label{NI}
M_0^2(x)=\frac{8}{\pi^2}
        \int_0^{\infty} K_0(2x\sinh t)\,dt, \quad M_1^2(x)=\frac{8}{\pi^2}
        \int_0^{\infty} \cosh(2t)\, K_0(2x\sinh t)\,dt,
\end{equation}
and using  the formula (see \cite[10.29.3]{dlmf} that
\begin{equation} \label{K0K1}
K_0'(x)=-K_1(x), \quad x>0,
\end{equation}
where in (\ref{NI}), (\ref{K0K1}),  $K_0(x)$ and $K_1(x)$ denote the second standard solutions to the Modified Bessel equation,
with $\nu =0,$ $1$, respectively.

By considering the asymptotic behavior of $K_0(x)$ and $K_1(x)$ for $0<x \ll 1$ \cite[10.30.2,10.30.3]{dlmf} and  $x \gg 1$ \cite[10.25.3]{dlmf}, it readily  follows that the formal differentiation of
 the (convergent) representations for $M_0(x)$ and $M_1(x)$ given in (\ref{NI}) is justified,
\begin{equation} \nonumber
\left(M^2_0(x)\right)'=   -\frac{8}{\pi^2}
        \int_0^{\infty} K_1(2x\sinh t)\cdot 2\sinh t\,dt, \quad x>0,
\end{equation}
\begin{equation} \nonumber
\left(M^2_1(x)\right)'=   -\frac{8}{\pi^2}
        \int_0^{\infty}\cosh(2t)\,  K_1(2x\sinh t)\cdot 2\sinh t\,dt, \quad x>0,
\end{equation}
as the integrals above are convergent uniformly in $x$, for $x>0.$
As the functions $\sinh(x)$, $\cosh(x)$ and $K_1(x)$ are  strictly positive for  $x>0$, the claim follows.
\end{proof}
To complete the proof of Lemma \ref{Dmonotone}, let us  recall that $M_0(x)>0$ for $x > 0$,  and  $\cos \big(\theta_0(\qp) - \p \big)<0$ and $\cos \big(\theta_0(\qm) - \p \big)>0$ by  Lemma \ref{signcosine}.
Hence the above claim implies that $\pm \frac{d}{dt}\mathcal{D}(q_\pm(t))>0$, and therefor $D'(t)>0$ in accordance with the definition of $D(t)$ in  (\ref{defD}). \end{proof}
This  completes the proof of Proposition \ref{existu_q}. \end{proof}
To return now and complete the proof of Theorem \ref{existence_r}, let  $q_+^{\sup}=(R_0 -\delta)/\epsilon)$ and let us consider the corresponding value of $q_0^{\sup}$.
Since  $q_+^{\sup} = \theta_1^{-1}(\theta_1(q_0^{\sup} - t + \pi/2))$ for some $t \in [-\pi/1, \pi/2]$ by (\ref{46}), we obtain that $q_+^{\sup}=q_0^{\sup} + O(1)$
and hence $r_0^{\sup}=r_+^{\sup} + O(\epsilon)$
as claimed earlier. \end{proof}
%

\subsection{Dimple solutions for (DQOP)} \label{subsecDR}
We now consider radial ``dimple'' solution in $\Omega=\mathcal{B}_{R_0}(\delta)$, with    $-1<u(r) < 1$ for $r \in (0, r_+)$ and $u(r) \equiv -1$ for $r\in [r_+, R_0),$
where $r_+$ reflects a free boundary with $0 <r_+ < R_0 - \delta.$
 As we saw  in Section \ref{subsect2.2}, $u(r)$ should satisfy
\begin{equation} \nonumber
\begin{split}
& u+\lambda=-\e^2(u_{rr}+\frac{1}{r}u_r), \quad  0<r<r_+,\\
& u(0) \in (-1,1] \ , \quad u(r_+)=-1,\\
& u_r(0)=0\ , \quad u_r(r_+)=0,
\end{split}\end{equation}
and  the equivalent mean mass condition implies that
\begin{equation} \label{dimple_mm}
\bar{u}=\frac{(2 r_0^2 - R_0^2)}{R_0^2}, \quad
\lambda = \frac{r_+^2  - 2 r_0^2}{r_+^2}, \quad r_0=\Bigl[\frac{1+ \bar{u}}{2}\Bigr]^{1/2} R_0,
\end{equation}
for some $0 < r_0 <r_+$, see (\ref{dimple_r})--(\ref{lambda_d}).
Thus $\lambda=\lambda(r_+, r_0)$ and $r_0=r_0(\bar{u}, R_0).$ We shall see that for $0<\epsilon \ll 1$, nontrivial radial dimple solutions exist for all $-1< \bar{u}=-1 + O(\epsilon)$ and $0<r_0 = O(\epsilon)$ sufficiently small; in particular, as $\epsilon \downarrow 0,$ $\bar{u} \downarrow -1,$ $r_0 \downarrow 0$, and $u(r;\epsilon) \downarrow -1$, $0<r<R_0.$

The parameter $r_+$, reflecting the location of the free boundary, is to be determined. The value $u(0) \in (-1,1]$ is
as an additional free parameter to be determined. Thus, we have three boundary conditions  and the equivalent mean mass constraint, $\bar{u}=(2 r_0^2 - R_0^2)/{R_0^2},$
as well as the range constraint on $u(0)$.
In solving (\ref{dimple_r})--(\ref{lambda_d}), we have four degrees of freedom, counting  the parameters $u(0)$ and $r_+,$ with  $\lambda=\lambda(r_+, r_0)$. Hence, from the ``count'' of the parameters, we would expect
the possible dimple solutions to be uniquely determined by
 $\bar{u}$ (or, equivalently, by $r_0$).

Setting
\begin{equation}\label{rescaling_d}
q=\frac{r}{\e},  \ \quad
     q_0=\frac{r_0}{\e}, \quad
    q_+=\frac{r_+}{\e}, \quad Q_0=\frac{R_0}{\e}, \quad \hbox{and} \quad v(q):=[u(r)+\lambda]|_{r=\epsilon q},
\end{equation}
    we get the following problem for $v(q)$,
\begin{equation}\label{dimple_v}
\begin{split}
& qv_{qq}+v_q+qv=0, \qquad 0<q< q_+,\\
& v(0)=u(0)+\lambda \in (\lambda  -1, \lambda +1], \quad  v(q_+)=-1+\lambda,\\
& v_q(0)=0, \quad v_q(q_+)=0,
\end{split}\end{equation}
where $\lambda= {q_+^2  - 2 q_0^2}/{q_+^2}$.              The equation in (\ref{dimple_v})  is a Bessel's equation of order zero, whose general solution is $v(q)=c_1 J_0(q)+c_2 Y_0(q),$  $c_1, \, c_2 \in R,$
and hence  $v_q(q)=-c_1 J_1(q)-c_2 Y_1(q)$).
Since $J_0(0)=1,$ $J_1(0)=0,$ and $Y_1(0) \ne 0,$
 the boundary conditions at $q= 0$ imply that $c_1=u(0)+\lambda,$
$c_2=0$. Thus
\begin{equation}\label{vq}
v(q)= (u(0)+\lambda) J_0(q),\quad
v_q(q)=-(u(0)+\lambda) J_1(q).
\end{equation}

The values of $J_0$  at its sequential minima are increasing  \cite[10.3,10.21,10.18]{dlmf}, and  $v_q(q_+)=0$ with $v(q_+)=-1 + \lambda,$ which corresponds to $u_r(r_+)=0,$ $u(r_+)=-1$. Hence
the range constraint,  $-1 \le u(r) \le 1,$ implies  that $q_+=\bar{q},$ where $\bar{q}$ corresponds
to the (unique) first positive zero of the function $J_1(q)$ \cite[10.21]{dlmf}. Thus \cite[Tables 9.1, 9.5]{AS},
\begin{equation}\label{q_0}
q_+= \bar{q} \approx 3.8  \quad \hbox{ with }\quad  J_0(q_+)=J_0(\bar{q})\approx -0.4.
\end{equation}

From (\ref{dimple_v})--(\ref{q_0}), we obtain that
\begin{equation}\label{vgg}
v(q_+)=v(\bar{q})=-1+\lambda = (u(0)+\lambda)J_0(\bar{q}),
\end{equation}
and using (\ref{dimple_mm}), (\ref{rescaling_d}), (\ref{vgg}),
\begin{equation} \label{lambda0}
\lambda=\frac{u(0)J_0(\bar{q})+1}{1-J_0(\bar{q})}=\frac{r_+^2 - 2 r_0^2}{r_+^2}=\frac{r_+^2-(1 + \bar{u})R_0^2}{r_+^2}. \end{equation}
From (\ref{q_0}), (\ref{lambda0}),
\begin{equation} \label{r0_d}
r_0=\epsilon \bar{q} \, \Bigl[ \frac{-(1+ u(0))J_0(\bar{q})}{2(1-J_0(\bar{q}))}\Bigr]^{1/2}, \quad \bar{u} = -1 - \epsilon^2 \,\frac{\bar{q}^2 (1 + u(0))J_0(\bar{q})}{R_0^2 (1-J_0(\bar{q}))},
\end{equation}
and recalling (\ref{rescaling_d}), (\ref{vq}), (\ref{q_0}), (\ref{lambda0}), (\ref{r0_d}),
\begin{equation} \label{dimple_s}
u(r) =  \frac{(u(0) +1) J_0({r}/{\epsilon}) -u(0)\,J_0(\bar{q})-1}{1 - J_0(\bar{q})}, \quad 0\le r \le r_+ = \epsilon \bar{q},
\end{equation}
with
\begin{equation} \label{dimple_s1}
 u(0)+1= \frac{(1+ \bar{u})R_0^2 (J_0(\bar{q})-1)}{\epsilon^2 \bar{q}^2 J_0(\bar{q})}.
\end{equation}
From (\ref{lambda0})--(\ref{dimple_s}), we see that there is a one-parameter family of dimple solutions which is uniquely determined by $\bar{u}$, or equivalently by $r_0$ or
by $u(0)$, with $u(0) \in [-1,1]$.

As $u(0) \downarrow -1$,
$$ \lambda \uparrow 1, \quad r_0 \downarrow 0, \quad \bar{u} \downarrow -1, \quad \hbox{and} \quad u(r) \downarrow -1, \quad \forall r \in [0, R_0].$$

As $u(0) \uparrow +1$,
$$\lambda \downarrow \frac{1+ J_0(\bar{q})}{1-J_0(\bar{q})}, \quad  r_0 \uparrow  \epsilon \bar{q} \, \Bigl[ \frac{-2 J_0(\bar{q})}{2(1-J_0(\bar{q}))}\Bigr]^{1/2},
\quad \bar{u} \uparrow -1 - \epsilon^2\frac{ 2 \bar{q}^2  J_0(\bar{q})}{R_0^2 (1-J_0(\bar{q}))},$$
and $u(r) \uparrow \Bigl[\frac{2 J_0({r}/{\epsilon}) -J_0(\bar{q})-1}{1 - J_0(\bar{q})}\Bigr]$  for  $r\in [0, r_0]$.
Note in particular that
$$\lim_{\epsilon \downarrow 0} \lim_{u(0) \downarrow -1} u(r; \epsilon, u(0)) = \left\{ \begin{array}l -1 \quad \forall r \in (0, R_0], \\ 1 \quad r=0, \end{array}
\right. \quad \lim_{\epsilon \downarrow 0} \lim_{u(0) \downarrow -1}r_0(\epsilon, u(0))=0.$$

Summarizing the results above,
\begin{theorem} \label{nrdimples}
Given $0<\epsilon \ll 1$ and $\mathcal{B}_{R_0}(\delta),$ with $R_0=O(1)$ and $\epsilon \ll \delta \ll 1$. There exists a unique radial dimple solution,
for any $\bar{u} \in (-1, -1 - 2 \epsilon \bar{q}/R_0)^{2} J_0(\bar{q})/(1-J_0(\bar{q}))$, where $\bar{q}$ denotes the first positive zero of $J_1(q)$.
\end{theorem}

\bigskip
Finally,  using (\ref{dimple_r}) and recalling (\ref{funcDQOP}), it is straightforward to verify that
\begin{equation} \label{E_dimple}
E(t):=  \frac{1}{\epsilon |\Omega|} \int_{\Omega} \{ (1-u^2) + \epsilon^2 |\nabla u|^2 \} \, dx   =\frac{\epsilon \bar{q}^2}{R_0^2} (1 - \lambda^2),
\end{equation}
where $\bar{q}$ and $\lambda$ are prescribed  in (\ref{q_0}) and (\ref{lambda0}), respectively.



\section{Connecting the Dynamic Problems}\label{section3}

At least on a somewhat superficial level, the attractor dynamics for both (SD) and (DQOP) can be seen to be similar in $2D$,
under the assumptions outlined in  Section \ref{Steadystates},  given $\mathcal{B}_{R_0}(\delta)$ with $0<\epsilon \ll \delta \ll  1$ and $R_0=O(1)$.
 The steady states for (SD) are given by
circular curves centered at the origin, as well as the possible translates of these circular curves  that lie within $\mathcal{B}_{R_0}(\delta)$ and the possible nonintersecting union of such curves that lie within
$\mathcal{B}_{R_0}(\delta)$. Similarly,     the steady states for (DQOP) are given radial solutions which are  either annular or dimple solutions, as well as their possible translates within $\mathcal{B}_{R_0}(\delta)$
and solutions obtain via composites of the above given the limitations of the domain $\mathcal{B}_{R_0}(\delta)$.
The equivalent mean mass condition for (DQOP) prescribes the mean mass, $\bar{u}$, in terms of an effective radius, $r_0$, with $0\le r_0 <R_0 -\delta.$
Accordingly, it is possible to identify a $1-1$ correspondence between the set of radial steady states of (SD), namely circular curves with radius $r_0$, with $0< r_0 < R_0 -\delta$,  and  the set of radial steady states
for (DQOP). For simplicity, we may limit our focus to the set of axi-symmetric steady states in both cases, leaving aside for the moment the
technical difficulties entailed in taking into account the somewhat larger class of steady states produced by translation. Also, we are neglecting possible radial ringed solutions, composed of concentric circular
curves for (DQ) and radially symmetric composite multiple transition (multi-annular solutions, possibly with a dimple solution at the origin).

\smallskip
Thus, if we can identify similar stability properties for both evolutions, we are well on our path to connecting the evolutions. The difficulty arises in considering stability for both evolutions
in similar functional analytic settings and in  a manner which permits both evolutions to be simultaneously tracked globally in time. We first outline briefly the perhaps easiest and  most direct approach,
which arises naturally  in view of
 the extant results in the literature for both evolutions, explaining some of the pitfalls in connecting the evolutions. Afterwards,  we outline some of the details pertaining to a more robust approach. The more robust approach is
 based on considering similar minimizing motion evolutionary descriptions for both evolutions, and making a step-by-step connection between the two motions via an appropriate lifting and projecting algorithm.

\subsection{Stability} \label{section3s}


With regard to the stability in the context of (SD) for circular curves, the results of Wheeler \cite{Wheeler2013} are useful. The theory there is based  on the   following  local existence theorem, which is paraphrased below:
\begin{theorem} \label{Wheeler_existence}
Suppose that $\Gamma_0: R  \rightarrow R^2$ is a periodic regular curve parametrised by arc-length of class $\mathcal{C}^2 \cap W^{2,2}$ with $||\kappa||_2< \infty.$ Then there exists a time $T\in (0, \infty]$ and a unique one-parameter family of immersions $\Gamma: R \times [0,T) \rightarrow R^2$ parametrized by arc-length satisfying (SD) such that $(i)$ $\Gamma(0)=\Gamma_0,$ $(ii)$ $\Gamma(\cdot,t)$ is of class $\mathcal{C}^{\infty}$ and periodic of period
$\mathcal{L}(\Gamma(\cdot, t))$ for every $t \in (0,T)$, and $(iii)$ $T$ is maximal.
\end{theorem}
 The regularity requirements on the initial data $\Gamma_0$ can be somewhat weakened,  \cite{Wheeler2013}.

The stability results \cite[Theorem 1.1]{Wheeler2013}, paraphrased below,   are formulated in terms of the normalized oscillation of curvature,
\begin{equation} \nonumber
\kappa_{osc}(\Gamma(\cdot,t)):=\mathcal{L}(\Gamma(\cdot,t))\int_{\Gamma} (\kappa - \bar{\kappa})^2 \, ds \quad\hbox{with}\quad \bar{\kappa}(\Gamma(\cdot,t))={\mathcal{L}}^{-1}(\Gamma(\cdot, t)) \int_{\Gamma} \kappa \, ds,
\end{equation}
 and the isoperimetric ratio, $\mathcal{I}(\Gamma(\cdot,t)):={\mathcal{L}}^2(\Gamma(\cdot, t))\,[4 \pi \mathcal{A}(\Gamma(\cdot,t))]^{-1}$, where  $\mathcal{A}(\Gamma(\cdot,t))$ denotes the area  enclosed by $\Gamma(\cdot,t)$.
 \smallskip
 \begin{theorem}  \label{Wheeler_stability}
 Suppose that $\Gamma_0: S^1 \rightarrow R^2$ is a regular smooth immersed closed curve with $\mathcal{A}(\Gamma_0)>0$ and $\int_{\Gamma_0} \kappa \, ds=2 \pi.$ There exists a constant $\kappa^\ast>0$ such that if
 \begin{equation} \nonumber
 \kappa_{osc}(\Gamma_0)< \kappa^\ast \quad \hbox{\, and \,} \quad \mathcal{I}(\Gamma_0) < exp\left(\frac{\kappa^\ast}{8 \pi^2}\right),
 \end{equation}
 then under (SD) evolution, $\Gamma: S^1 \times [0,T) \rightarrow R^2$ with $\Gamma_0$ as initial data exists for all time and converges exponentially fast to a round circle with radius $\sqrt{\mathcal{A}(\Gamma_0)/\pi}$.
 \end{theorem}
 \smallskip
 The results in \cite{Wheeler2013} are quite pleasing. However, one cannot conclude directly from either of these theorems
that if $\Gamma_0 \subset \mathcal{B}_{R_0}(\delta),$ then
$\Gamma(\cdot,t) \subset \mathcal{B}_{R_0}(\delta)$ for $t \in (0,T)$.

\smallskip
With regard to stability in the (DQOP) context, suppose we wish to demonstrate stability for an annular solution located far from the $\mathcal{B}_{R_0}(\delta)$ boundary.
For simplicity, let us consider the stability of an annular solution centered at the origin with $r_0=O(1).$ We know for (SD) that the encompassed area is maintained, but we
do not know off hand that the center of mass does not move. For (DQOP) we similarly know that mass is conserved, but we do not know that the center of mass is time invariant.
So, if we are not overly concerned with maintaining the $\mathcal{B}_{R_0}(\delta)$ structure, a reasonable approach is to to consider zero mass perturbations, making use of the $H^{-1}$ gradient structure.
  Within this context establishing a spectral gap should be straightforward, for zero mass perturbations modulo translations of the center of mass. This would enable us to prove stability of the annular solutions, modulo
  translation, in analogy with the (SD) results above.

\smallskip
In Section \ref{subsecDR}, there exist
 $O(\epsilon)$ energy dimple solutions which equal $-1$ except on a circular region with $O(\epsilon^2)$ area. Clearly a ``stray'' translate of a dimple
(or rather a translate of that part of the dimple solution which differs from $-1$) could be incorporated into the $O(1)$ region where a generic annular solutions equal
$-1$, with a small alteration in the radii of the annular solution to accommodate the additional mass. Such a small ``droplet-like'' perturbation would constitute a small energy perturbation, though not covered by the discussion above, as they are not zero mass perturbations. As such perturbations are natural to consider, we remark here that it is possible to construct a sequence of energy lowering and mass preserving perturbations, which
allow the dimple to lose height and to transfer away volume (mass).  Details to be published elsewhere, together with the (DQOP) stability results described above.

\subsection{Minimizing motions} \label{section3m}
\bigskip
\textbf{Minimizing motion evolution formulation for (SD).}
In Fonseca et.al. \cite{FonsecaFLM},  De Giorgi's minimizing movement approach is implemented within a framework  with  $H^{-1}$ gradient flow structure, to prove
short time existence, uniqueness, and regularity  for the motion of an elastic thin film which evolves by anisotropic surface diffusion.
Their approach yields, as a subcase, a proof of short time existence, uniqueness, and regularity for a spatially period $1D$ curve prescribable
by the graph of a function, $\Gamma=\{(x, h(x): 0<x<b \}$, with $b \in (0, \infty).$

More specifically, starting with    $b$ periodic data initial, $h_0 \in H^2_{loc}(R),$
a sequence of approximants $h_{i,N}$ are defined inductively, for $T>0$, $N \in \mathcal{N}$, $i=1, \ldots, N,$ as a minimizer of $E(h) + \frac{1}{2\tau} d^2(h, h_{i-1, N}),$
where $E(h)$ denotes the system energy, $\tau=T/N$, and $d$ measures the $H^{-1}$ distance between $h$ and $h_{i-1, N}$.
Their choice for $d$ is based on  the following $H^{-1}$ norm for curves $\Gamma$,
$$ || f||_{H^{-1}(\Gamma)} := \sup_{|| \psi||_{H^1(\Gamma)=1}} \int_{\Gamma} f \psi d\mathcal{H}^1(z),$$
which  can be expressed as
\begin{equation} \nonumber
 || f||^2_{H^{-1}(\Gamma)} = \int_{\Gamma} \left( \Phi(z) - \frac{1}{|\Gamma|}\int_{\Gamma} \Phi d\mathcal{H}^1\right)^2 \, d\mathcal{H}^1(z) + \left( \int_{\Gamma} f \, d\mathcal{H}^1 \right)^2,
\end{equation}
where $\Phi(z):=\int_{\Gamma(z_0, z)} f(w) \, d\mathcal{H}^1(w)$,  $z_0=(0, h_{i-1, N}(0))$ and $\Gamma(z_0, z)$ denotes the arc of $\Gamma$ connecting $z_0$ with $z$. Accordingly
they base their scheme  on the $(H^{-1})^2$ penalization
\begin{equation} \label{penalization}
\int_{\Gamma} \left( \int_{\Gamma(z_0, z)} f(w) \, d\mathcal{H}^1(w)\right)^2 \, d\mathcal{H}^1(z),\end{equation}
for $f=h_i - h_{i-1},$ with two constraints reflecting zero mean and periodicity:
\begin{equation} \label{constraints}
\int_{\Gamma} f\, d\mathcal{H}^1 =0, \quad \int_{\Gamma} \int_{\Gamma(z_0, z)} f(w) \, d\mathcal{H}^1(w) \, d\mathcal{H}^1(z)=0.
\end{equation}

To implement their approach in our context, an extension of their approach is needed for closed imbedded planar curves, $\Gamma$.
The resultant penalization with constraints is similar to (\ref{penalization}),(\ref{constraints}), but incorporates an orientational weight.

\bigskip\par\noindent
\textbf{Minimizing motion evolution formulation for (DQOP).}
In \cite{LisiniMS} a minimizing movement scheme  is developed for proving the existence of weak  solutions for a class of degenerate
parabolic equations of fourth order, which includes (DQOP) given the assumptions outlined in Section \ref{sec1}.  This class of evolution equations are shown to
correspond to a gradient flow with respect to a Wasserstein-like transport metric, $W_m$, and the weak solutions are obtained via a scheme based on curves of maximal slope.

The  metric, $W_m$,
which was shown in \cite{Dolbeault} to constitute a genuine metric,
is neither the $L^2$-Wasserstein distance not a flat Hilbertian metric; rather it corresponds to a metric tensor based on the following construction:
given a tangential vector ${v}:=\partial_s \rho(0)$ to a smooth curve $\rho: (-\e, \e) \rightarrow L^1(\Omega)$ of strictly positive densities $\rho(s)$
at $\rho_0=\rho(0)$, it assigns the length
\begin{equation} \label{Lisini8}
||v||^2 = \int_\Omega |D\psi(x)|^2 \, m(\rho_0(x))\, dx, \hbox{\, with \,} -\nabla \cdot (m(\rho_0(x))\nabla \psi(x))=v, \quad x \in \Omega,
\end{equation}
with variational boundary conditions on $\partial \Omega$. Within the context of our framework, $\Omega=\mathcal{B}_{R_0}(\delta)$, as discussed
in Section \ref{sec1}, and
the ambient space for the scheme is the metric space $(X(\Omega), W_m)$ where
\begin{equation}\nonumber
X(\Omega):=\Bigl\{ u \in L^1(\Omega)\, | \, -1 \le u \le 1 \quad a.e.\quad x \in \Omega, \quad \int_\Omega u \, dx =\bar{u}\Bigr\}.
\end{equation}

The minimizing movement scheme starts with initial conditions $u_0 \in X(\Omega)$ such that $E[u_0] < \infty$, and approximants are defined by setting
\begin{equation} \nonumber
u^0_{\tau}:=u_0, \quad u_{\tau}^{n+1}:-\arg\min \Phi_{\tau}^n \in X(\Omega), \quad \Phi_{\tau}^n(\nu):=\frac{1}{2 \tau} W_m(u_{\tau}^n,\nu)^2 + E[\nu],
\end{equation}
with $E[u]:= \infty$ for $u \not\in H^1(\Omega)$. Piecewise constant interpolation is then used to define an approximation $\tilde{u}_{\tau}(t):[0, \infty) \rightarrow X(\Omega).$

\bigskip\par\noindent
\textbf{Lifting and projecting: a Hilbert expansion approach.} As we have seen, there exist minimizing motion schemes
for both (SD) and (DQOP). To establish a rigorous connection between the two evolutions, (SD) and (DQOP), it makes sense to
consider similar time steps, $\tau=T/N$, for both minimizing motions. It is also reasonable to consider similar initial conditions, perhaps similarly
perturbed steady state solutions to (SD) and (DQOP); for simplicity, we might consider  a circle and an annular solution, respectively, with  equivalent mean mass and
both centered at the origin.
 We want to compare the results of the minimizing motion schemes and to demonstrate that
they yield  the same motion in the limit as $\epsilon \rightarrow 0$ and $\tau \rightarrow 0$. Here  the difficulty is that at each time step
the motion
for (SD) is describable in terms of curves $\Gamma_t$ belonging to  $\mathcal{M}$,  the set of smooth simple closed curves, and the motion for (DQOP) is described in terms of functions $u(\cdot, t)$ defined for all $x \in \mathcal{B}_{R_0}(\delta).$  Modulo issues of regularity and structure,    the zero  level set of  $u(\cdot, t)$ yields a ``projection'' of $u(\cdot, t)$ onto the set of curves $\mathcal{M}$.
The matter of ``lifting'' $\Gamma_t$ to obtain   $u(\cdot, t)$ is more delicate. In the context of connecting the Cahn-Hilliard equation with the Mullins-Sekerka problem, Carlen, Carvalho \& Orlandi \cite{CCO}
 used a Hilbert expansion approach to  construct a globally defined function $u(\cdot, t)$ from a curve $\Gamma_t \in \mathcal{M}$.  They constructed an approximation for
  $u(\cdot, t)$ based on  three types of terms, $(i)$ terms depending on $\Gamma_t$, $(ii)$ terms reflecting  local corrections near $\Gamma_t$, and
$(iii)$ terms reflecting long range corrections.

%
%



%


%
\bigskip\par\noindent
\textbf{Acknowledgments}
The authors would like to acknowledge support from the Israel Science Foundation (Grant No. 1200/16).

\end{document}